\def\ReportVersion{}

\ifdefined\ReportVersion
\documentclass[10pt]{article}
\else
\documentclass{sig-alternate}
\fi

\usepackage{graphics}
\usepackage{amssymb}
\usepackage{amsmath}

\ifdefined\ReportVersion
\usepackage{amsthm}
\fi

\usepackage{mathtools}
\usepackage{rtsched}
\usepackage{color}
\usepackage{url}

\usepackage[shortcuts]{extdash}

\ifdefined\ReportVersion
\newcommand{\reportOnly}[1]{#1}
\newcommand{\paperOnly}[1]{}
\else
\newcommand{\reportOnly}[1]{}
\newcommand{\paperOnly}[1]{#1}
\fi

\newcommand{\inactive}{{\sf Inactive}}
\newcommand{\activeContending}{{\sf ActiveContending}}
\newcommand{\activeNonContending}{{\sf ActiveNonContending}}
\newcommand{\recharging}{{\sf Recharging}}

\renewcommand{\d}[1]{\ensuremath{\operatorname{d}\!{#1}}}

\ifdefined\CommentedVersion
\newcommand{\gl}[1]{{\color{blue} \textbf{Peppe}: #1}}
\newcommand{\la}[1]{{\color{red} \textbf{Luca}: #1}}
\newcommand{\ys}[1]{{\color{green} \textbf{Youcheng}: #1}}
\else
\newcommand{\gl}[1]{}
\newcommand{\la}[1]{}
\newcommand{\ys}[1]{}
\fi

\newtheorem{theorem}{Theorem}
\newtheorem{lemma}{Lemma}
\newtheorem{definition}{Definition}

\paperOnly{
\title{Multicore CPU Reclaiming: Parallel or Sequential?}
}
\reportOnly{
  \title{TECHNICAL REPORT \\ Parallel and sequential reclaiming in multicore
    real-time global scheduling}
}

 \ifdefined\ReportVersion
 \author{Luca Abeni\\
   University of Trento, Italy\\
   \texttt{luca.abeni@unitn.it}
   \and 
   Giuseppe Lipari \\
   Univ. Lille, CNRS, Centrale Lille, UMR 9189 - CRIStAL, France\\
   \texttt{giuseppe.lipari@univ-lille1.fr}
   \and
   Andrea Parri and Youcheng Sun\\
   Scuola Superiore Sant'Anna, Pisa, Italy\\
   \texttt{\{a.parri, y.sun\}@sssup.it}
 }
 \else
 \numberofauthors{4}
 \author{
 \alignauthor
 Luca Abeni\\
       \affaddr{University of Trento}\\
       \email{luca.abeni@unitn.it}
 \alignauthor
 Giuseppe Lipari\\
       \affaddr{Univ. Lille, CNRS, Centrale Lille, UMR 9189 - CRIStAL}\\
       \email{giuseppe.lipari@univ-lille1.fr}
 \alignauthor
 Andrea Parri, Youcheng Sun\\
       \affaddr{Scuola Superiore Sant'Anna}\\
       \email{\{a.parri, y.sun\}@sssup.it}
 }
 \fi

\begin{document}
\ifdefined\ReportVersion
\else
\CopyrightYear{2016}
\setcopyright{othergov}
\conferenceinfo{SAC 2016,}{April 04-08, 2016, Pisa, Italy}
\isbn{978-1-4503-3739-7/16/04}\acmPrice{\$15.00}
\doi{http://dx.doi.org/10.1145/2851613.2851743}
\fi

\maketitle

\begin{abstract}
  When integrating hard, soft and non-real-time tasks in general
  purpose operating systems, it is necessary to provide temporal
  isolation 
  so that the timing properties of one task do not depend on the
  behaviour of the others.  However, strict budget enforcement can
  lead to inefficient use of the computational resources in the
  presence of tasks with variable workload. Many resource reclaiming
  algorithms have been proposed in the literature for single processor
  scheduling, but not enough work exists for global scheduling in
  multiprocessor systems. In this \paperOnly{paper}\reportOnly{report}
  we propose two reclaiming algorithms for multiprocessor global
  scheduling and we prove their correctness. \paperOnly{We also
    present their implementation in the Linux kernel and we compare
    their performance on synthetic experiments.}
\end{abstract}

\ifdefined\ReportVersion
\relax
\else



\fi

\section{Introduction}
\label{sec:introduction}

The Resource Reservation Framework
\cite{rajkumar1997resource,abeni1998integrating} is an effective
technique to integrate the scheduling of real-time tasks in
general-purpose systems, as demonstrated by the fact that it has been
recently implemented in the Linux kernel~\cite{SPE:SPE2335}. One of
the most important properties provided by resource reservations is
{\em temporal isolation}: the worst-case performance of a task does
not depend on the temporal behaviour of the other tasks running in the
system. This property can be enforced by limiting the amount of time
for which each task can execute in a given
period. 

In some situations, a strict enforcement of the executed runtime (as
done by the hard reservation mechanism that is currently implemented
in the Linux kernel) can be problematic for tasks characterized by
highly-variable, or difficult to predict, execution times: allocating
the budget based on the task' Worst Case Execution Time (WCET) can
result in a waste of computational resources; on the other hand,
allocating it based on a value smaller than the WCET can cause a
certain number of deadline misses.
These issues can be addressed by using a proper {\em CPU reclaiming
  mechanism}, which allows tasks to execute for more than their
reserved time {\bf if spare CPU time is available} and if this
over-execution does not break the guarantees of other real-time tasks.

While many algorithms
(e.g.,~\cite{nogueira2007capacity,lin2005improving,caccamo2000capacity,conf/ecrts/LipariB00})
have been developed for reclaiming CPU time in single-processor
systems, the problem of reclaiming CPU time in multiprocessor (or
multicore) systems has been investigated less. Most of the existing
reclaiming algorithms (see~\cite{lin2005improving} for a summary of
some commonly used techniques) are based on keeping track of the
amount of execution time reserved to some tasks, but not used by them,
and by distributing it between the various active tasks. In a
multiprocessor system, this idea can be extended in two different
ways:
\begin{enumerate}
\item by considering a {\em global} variable that keeps track of the
  execution time {\em not used by all the tasks in the system}
  (without considering the CPUs/cores on which the tasks execute), and
  by distributing such an unused execution time to the tasks.  This
  approach will be referred to as {\em parallel reclaiming} in this
  paper, because the execution time not used by one single task can be
  distributed to multiple tasks that execute in parallel on different
  CPUs/cores;
\item by considering multiple per-CPU/core (per-runqueue, in the Linux
  kernel \emph{slang}) variables each representing the unused bandwidth
  that can be distributed to the tasks executing on the corresponding
  CPU/core. This approach will be referred to as {\em sequential
    reclaiming} in this paper, because 
  the execution time not used by one single task is associated to a
  CPU/core, and cannot be distributed to multiple tasks executing
  simultaneously.
\end{enumerate}

This paper compares the two mentioned approaches by extending the GRUB
(Greedy Reclamation of Unused Bandwidth) \cite{conf/ecrts/LipariB00}
reclaiming algorithm to support multiple processors according to
sequential reclaiming and parallel reclaiming.  The comparison is
performed both from the theoretical point of view (by formally
analysing the schedulability of the obtained algorithm) and by running
experiments on a real implementation of this extension, named M-GRUB.
Such implementation of M-GRUB reclaiming (that can do either parallel or
sequential reclaiming) is based on the Linux kernel and extends the
{\tt SCHED\_DEADLINE} scheduling policy.

The \paperOnly{paper}\reportOnly{report} is organised as follows: in
Section~\ref{sec:related-work} we recall the related work. In
Section~\ref{sec:sys-model} we present our system model and introduce
the definitions and concepts used in the paper. The algorithms and
admission tests used as a starting point for this work are then
presented in Section~\ref{sec:background}.  The two proposed
reclaiming rules are described in
Section~\ref{sec:reclaiming-rule}. \reportOnly{In Section
  \ref{sec:proof-correctness} we prove the correctness of the two
  algorithms.} \paperOnly{In Section~\ref{sec:implementation} we
  discuss the implementation details and in
  Section~\ref{sec:experiments} we present the results of our
  experiments.} Finally, in Section~\ref{sec:conclusions} we present
our conclusions.

\section{Related work}
\label{sec:related-work}

The problem of reclaiming unused capacity for resource reservation
algorithms has been mainly addressed in the context of single
processor scheduling.

The CASH (CApacity SHaring) algorithm \cite{caccamo2000capacity} uses
a queue of unused budgets (also called \emph{capacities}) that is
exploited by the executing tasks. However, the CASH algorithm is only
useful for periodic tasks. Lin and Brandt proposed BACKSLASH
\cite{lin2005improving}, a mechanism based on capacities that
integrates four different principles for slack reclaiming. Similar
techniques still based on capacities are used by Nogueira and
Pinho~\cite{nogueira2007capacity}. 

The GRUB algorithm~\cite{conf/ecrts/LipariB00} modifies the rates at
which servers' budgets are decreased so to take advantage of free
bandwidth. The algorithm can be also used by aperiodic tasks. We
present the GRUB algorithm in Section \ref{sec:background} as it is
used as a basis for our multiprocessor reclaiming schemes. For fixed
priority scheduling, Bernat et al. proposed to reconsider past
execution so to take advantage of the executing slack
\cite{bernat2004rewriting}. 

Reclaiming CPU time in multiprocessor systems is more difficult
(especially if global scheduling is used), as shown by some previous
work~\cite{Bra07} that ends up imposing strict constraints
on the distribution of spare budgets
to avoid 
compromising timing isolation: spare CPU time can only be donated by
hard real-time tasks to soft real-time tasks -- which are scheduled in
background respect to hard tasks -- and reservations must be properly
dimensioned.  

To the authors' best knowledge, the only previous algorithm that
explicitly supports CPU reclaiming on all the real-time tasks running
on multiple processors without imposing additional constraints (and
has been formally proved to be correct) is
M-CASH~\cite{pellizzoni2008m}. It is an extension of the CASH
algorithm to the multiprocessor case, which additionally includes a
rule for reclaiming unused bandwidth by aperiodic tasks.  The
algorithm uses the utilisation based test by Goossens, Funk and Baruah
\cite{goossens2003priority} as a base schedulability test for the
servers. It distinguishes two kinds of servers: servers for periodic
tasks (whose utilisation is reclaimed using capacity-based mechanism)
and servers for aperiodic tasks, whose bandwidth is reclaimed with a
technique similar to the parallel reclaiming that we propose in
Section~\ref{sec:parallel-recl}. However, M-CASH has never been
implemented in a real OS kernel. On the other hand, the GRUB
algorithm~\cite{conf/ecrts/LipariB00} has been implemented in the
Linux kernel~\cite{Abe14RTLWS}, after extending the algorithm to
support multiple CPUs, but the multiprocessor extensions used in this
implementation have not been formally analysed nor validated from a
theoretical point of view.

\section{System model and definitions}
\label{sec:sys-model}

We consider a set of~$n$ real-time tasks~$\tau_i$ scheduled by a set
of~$n$ servers~$S_i$ ($i = 1, \dots, n$).

A real-time \emph{task} $\tau_i$ is a (possibly infinite) sequence of
jobs $J_{i,k}$: each job has an arrival time $a_{i,k}$, a computation
time $c_{i,k}$ and a deadline $d_{i,k}$. {\em Periodic} real-time
tasks are characterised by a period $T_i$, and their arrival time can
be computed as $a_{i,k+1} = a_{i,k} + T_i$. {\em Sporadic} real-time
tasks wait for external events with a minimum inter-arrival time, also
called $T_i$, so $a_{i,k+1} \geq a_{i,k} + T_i$.
Periodic and sporadic tasks are usually associated a relative
deadline $D_i$ such that $d_{i,k} = a_{i,k} + D_i$.




A \emph{server} is an abstract entity used by the scheduler to reserve
a fraction of CPU-time to a task. Each server $S_i$ is characterised
by the following parameters: $P_i$ is the \emph{server period} and it
represents the granularity of the reservation; $U_i$ is the fraction
of reserved CPU-time, also called \emph{utilisation} factor or
\emph{bandwidth}. In each period, a server is reserved at least a
\emph{maximum budget}, or \emph{runtime}, $Q_i = U_i P_i$.

The execution platform consists of $m$ identical processors (Symmetric
Multiprocessor Platform, or SMP).
In this paper we use the Global Earliest Deadline
First (G\-/EDF) scheduling algorithm: all the tasks are ordered by
increasing deadlines of the servers, and the $m$ active tasks
with the earliest deadlines are scheduled for execution on the $m$
CPUs.

The logical priority queue of G-EDF is implemented in Linux by a set
of \emph{runqueues}, one per each CPU/core, and some accessory data
structures for making sure that the $m$ highest-priority jobs are
executed at each instant (see \cite{lelli2011efficient} for a
description of the implementation).

\section{Background}
\label{sec:background}

In this section we first recall the Constant Bandwidth Server (CBS)
algorithm~\cite{abeni1998integrating,Baruah2002} for both single and
multiprocessor systems. We then recall the GRUB
algorithm~\cite{conf/ecrts/LipariB00}, an extension of the
CBS. Finally, we present two schedulability tests for Global EDF.



\subsection{CBS and GRUB}
\label{sec:cbs-grub}

As anticipated in Section~\ref{sec:sys-model}, each server is
characterised by a period $P_i$, a bandwidth $U_i$ and a maximum
budget $Q_i = U_i P_i$.  In addition, each server maintains the
following dynamic variables: the \emph{server deadline} $d_i$,
denoting at every instant the server priority, and the \emph{server
  budget} $q_i$, indicating the remaining computation time allowed in
the current server period.

At time $t$, a server can be in one of the following states:
\activeContending, if there is some job of the served task that has
not yet completed; \activeNonContending, if all jobs of the served
task have completed, but the server has already consumed all the
available bandwidth (see the transition rules below for a
characterisation of this state); \inactive, if all jobs of the
server's task have completed and the server bandwidth can be reclaimed
(see the transition rules below), and \recharging, if the server has
jobs to execute, but the budget is currently exhausted and needs to be
recharged (this state is generally known as ``throttled'' in the Linux
kernel, or ``depleted'' in the real-time literature).

\reportOnly{
  \begin{figure}[tbh]
    \begin{center}
      \includegraphics[width=.75\columnwidth]{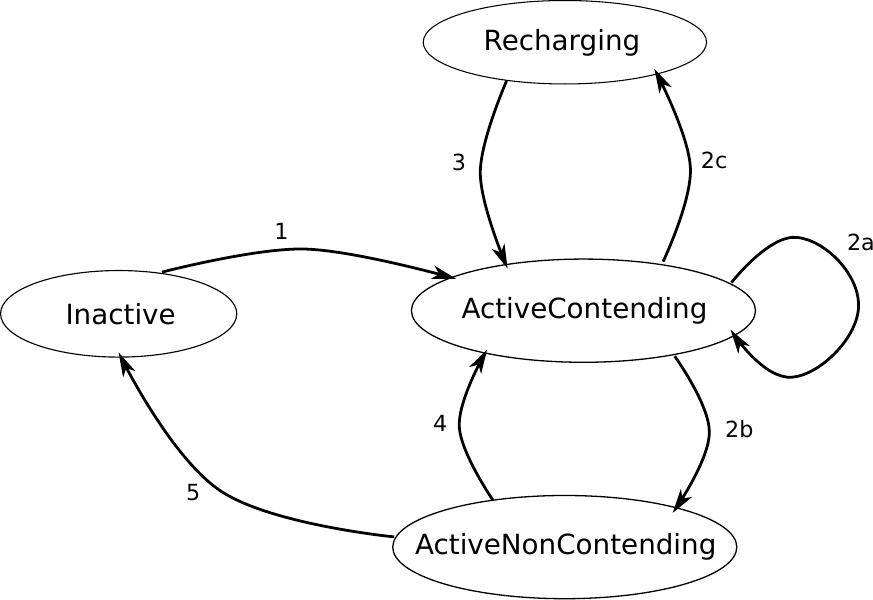}
    \end{center}
    \caption{State transition diagram for CBS and GRUB.}
    \label{fig:dia-mgrub}
  \end{figure}

  The states and the corresponding transitions are shown in Figure
  \ref{fig:dia-mgrub}.
}

The EDF algorithm chooses for execution the $m$ tasks with the earliest
server deadlines among the \activeContending\ servers.
Initially, all servers are in the \inactive\ state and their state
change according to the following rules:
\begin{enumerate}
\item[1.] When a job of a task arrives at time $t$, if the
  corresponding server is \inactive, it goes in \activeContending\ and
  its budget and deadlines are modified as: $q_i \leftarrow U_i P_i$
  and $d_i \leftarrow t + P_i$

\item[2a.] When a job of $S_i$ completes and there is another job
  ready to be executed, the server remains in \activeContending\ with
  all its variables unchanged;

\item[2b.] When a job of $S_i$ completes, and there is no other job
  ready to be executed, the server goes in \activeNonContending.

\item[2c.] If at some time $t$ the budget $q_i$ is exhausted, the
  server moves to state \recharging, and it is removed from the ready
  queue. The corresponding task is suspended and a new server is
  executed.

\item[3.] When $t = d_i$, the server variables are updated as $d_i
  \leftarrow d_i + P_i$ and $q_i \leftarrow U_i P_i$. The server is
  inserted in the ready queue and the scheduler is called to select
  the earliest deadline server, hence a context switch may happen.

\item[4.] If a new job arrives while the server is in
  \activeNonContending, the server moves to \activeContending\ without
  changing its variables.

\item[5.] A server remains in \activeNonContending\ only
  if $q_i < (d_i - t) U_i$.
  When $q_i \geq (d_i - t)U_i$ the server moves to \inactive.
\end{enumerate}

Only servers that are \activeContending\ can be selected for execution
by the EDF scheduler. If $S_i$ does not execute, its budget is not changed.
When $S_i$ is executing, its budget is updated as $\d{q_i} = -\d{t}$.

When serving a task, a
server generates a set of \emph{server jobs}, each one with an arrival
time, an execution time and a deadline as assigned by the algorithm's
rules. For example, when the server at time $t$ moves from \inactive\
to \activeContending\, a new server job is generated with arrival time
equal to $t$, deadline equal to $d_i = t + P_i$, and worst-case
computation time equal to $Q_i$. A similar thing happens when the
server moves from \recharging\ to \activeContending, and so on.

We say that a server is \emph{schedulable} if every server job can
execute the budget $Q_i$ before the corresponding server job deadline.
It can be proved that the \emph{demand bound function}~$\mathsf{dbf}$
(see \cite{baruah1990algorithms} for a definition) generated by the
server jobs of~$S_i$ is bounded from above as follows: $dbf(t) \leq
U_i t$ for each~$t$. Hence, for single processor systems it is
possible to use the utilisation test of EDF as an admission control
test, i.e.,
\begin{equation}
  \label{eq:sched-test-single}
   \sum_{i=1}^n U_i \leq 1.
\end{equation}

\reportOnly{
  We recall the following two fundamental results.
  \begin{theorem}[Theorem~1 in \cite{abeni1998integrating}]
    \label{thm:temporal-isolation}
    If Equation (\ref{eq:sched-test-single}) holds, then all server jobs
    meet their deadlines (i.e., all servers are schedulable).
  \end{theorem}
  
  \begin{theorem}[Lemma~1 in \cite{abeni1998integrating}]
    \label{thm:hard-schedulability}
    Assuming that Equation (\ref{eq:sched-test-single})
    holds, if a periodic (sporadic) hard real\-/time tasks with WCET $C_i$
    and period (minimum inter\-/arrival time) $T_i$ is assigned a server
    with budget $Q_i \geq C_i$ and period $P_i \leq T_i$, then the task
    will meet all its deadlines.
\end{theorem}
}

\paperOnly{
  It has been proved that if Equation (\ref{eq:sched-test-single})
  holds, then all servers are schedulable (i.e. all servers jobs will
  complete before their scheduling deadlines). Based on this result, it
  is possible to guarantee the respect of the deadlines of a task by
  setting $P_i \leq T_i$ and $Q_i \geq C_i$ (see the original
  paper~\cite{abeni1998integrating} for a more complete description).
}
  

The CBS algorithm has been extended to multiprocessor global
scheduling in~\cite{Baruah2002}. The authors prove the temporal
isolation and the hard schedulability properties of the algorithm when
using the schedulability test of Goossens, Funk and Baruah
\cite{goossens2003priority}, which we will recall next.

The GRUB algorithm~\cite{conf/ecrts/LipariB00} extends the CBS
algorithm by enabling the reclaiming of unused bandwidth, while
preserving the schedulability of the served tasks. The main difference
between the CBS and GRUB algorithms is the rule for updating the
budget. In the original CBS algorithm, the server budget is updated as
$\d{q_i} = -\d{t}$, independently of the status of the other servers.
To reclaim the excess bandwidth, GRUB maintains one additional global
variable $U_{\mathsf{act}}$, the total utilisation of all active
servers
\[
  U_{\mathsf{act}} \coloneqq \sum_{S_i \notin \mbox{\inactive}} U_i
\]
%
%
and uses it to update the budget $q_i$ as
\begin{equation}
  \label{eq:dq-grub}
  \cfrac{\d{q_i}}{\d{t}}  = -(1 - (U_{\mathsf{sys}} - U_{\mathsf{act}}))
\end{equation}
where $U_{\mathsf{sys}}$ is the utilisation that the system reserves
to the set of all servers.  As in the original CBS algorithm, the
budget is not updated when the server is not executing. The executing
server gets all the free bandwidth in a \emph{greedy} manner, hence
the name of the algorithm.  The GRUB algorithm preserves the Temporal
Isolation and Hard Schedulability properties of the CBS
\cite{lipari2000resource}.

\subsection{Admission control tests}
\label{sec:admiss-contr-tests}

When using the CBS or the GRUB algorithm it is important to run an
admission test to check if all the servers' deadlines are
respected. In single processor systems, the utilisation based test of
Equation (\ref{eq:sched-test-single}) is used for both the CBS and the
GRUB algorithm. We now present two different schedulability tests for
the multiprocessor case: an utilisation-based schedulability test for
G-EDF by Goossens, Funk and Baruah~\cite{goossens2003priority}
(referred to as GFB in this paper), and an interference-based
schedulability test for G-EDF proposed by Bertogna, Cirinei and
Lipari~\cite{bertogna2005improved} (referred to as BCL in this paper).

GFB and BCL are not 
the most advanced tests in the literature: in particular, as discussed
in \cite{Bertogna11}, more effective tests -- i.e. tests that can
admit a larger number of task sets -- are now available.  The reason
we chose these two in particular is their low complexity (so they can
be used as on-line admission tests), and the fact that currently we
are able to prove the correctness of the reclaiming rules with respect
to these two tests in particular. In fact, we need to guarantee that
the temporal isolation property continues to hold even when some
budget is donated by one server to the other ones according to some
reclaiming rule.

\paperOnly{At the time of preparation of this paper, we have formally
  proved the correctness of the two reclaiming rules proposed in
  Section \ref{sec:reclaiming-rule} with respect to the GFB and the
  BCL test -- the proofs are not reported here, due to space
  constraints, and can be found in a separate Technical
  Report~\cite{tech-report-mgrub}.}  \reportOnly{Currently, we have
  formally proved the correctness of the two reclaiming rules propose
  in Section \ref{sec:reclaiming-rule} with respect to the GFB and the
  BCL test -- see Section \ref{sec:proof-correctness}.}  Using some
other, more effective, admission control test may be unsafe, hence,
for the moment we restrict our attention to GFB and BCL.

The GFB test is based on the notion of uniform multiprocessor platform,
and it allows to check the schedulability of a task set based on its
utilisation.

\reportOnly{
  \begin{definition}
    A uniform multiprocessor platform is comprised of several
    processors. Each processor $p$ is characterised by its speed
    $\mathsf{speed}(p)$, with the interpretation that a job that
    executes on processor $p$ for $t$ times units, completes
    $\mathsf{speed}(p) \cdot t$ units of execution. Let $\pi$ denote the
    uniform multiprocessor platform. We introduce the following
    notation:
    \begin{align}
      \mathsf{MaxSpeed}_\pi &\coloneqq \max_{p \in \pi} \{ \mathsf{speed}(p) \} \\
      \mathsf{TotSpeed}_\pi &\coloneqq \sum_{p \in \pi} \mathsf{speed}(p).
    \end{align}
  \end{definition}
  Given a set of periodic (sporadic) tasks, we can always build a
  uniform multiprocessor platform on which the task set is schedulable,
  by providing one processor of speed $U_i$ per each task. This uniform
  multiprocessor platform will have $\mathsf{MaxSpeed}_\pi =
  \max_{i=1}^n \{ U_i \}$ and $\mathsf{TotSpeed}_\pi = \sum_{i=1}^n
  U_i$.
  
  \begin{definition}
    Let $I$ denote any set of jobs, and $\pi$ any uniform multiprocessor
    platform. For any scheduling algorithm A, and time instant $t \geq
    0$, let $W(A, \pi, I, t)$ denote the amount of work done by
    algorithm A on jobs of I over the interval $[0,t)$ while executing
    on $\pi$.
  \end{definition}
  
  The most general form of the test is given by the following Lemma.
  
  \begin{lemma}[Lemma~$1$ in~\cite{goossens2003priority}.]
    \label{lm-uniform-workload}
    Let $\pi$ denote a uniform multiprocessor platform with cumulative
    processor capacity $\mathsf{TotSpeed}_\pi$ and in which the fastest
    processor has speed $\mathsf{MaxSpeed}_\pi < 1$. Let $\pi'$ denote
    an identical multiprocessor platform comprised of $m$ processors of
    unit speed. Let A denote any uniform multiprocessor algorithm, and
    A' any work conserving scheduling algorithm on $\pi$. If the
    following condition is satisfied:
    \begin{equation}
      \label{eq:uniform-basic}
      m > \frac{\mathsf{TotSpeed}_\pi - \mathsf{MaxSpeed}_\pi}{1 - \mathsf{MaxSpeed}_\pi}
    \end{equation}
    then for any collection of jobs $I$ and any time instant $t\geq 0$,
    \begin{equation}
      \label{eq:uniform-workload}
      W(A',\pi',I,t) \geq W(A,\pi,I,t)
    \end{equation}
  \end{lemma}
  
  The following theorem specializes the result of Lemma
  \ref{lm-uniform-workload} for G-EDF.
  \begin{theorem}[Theorem~$3$ in~\cite{goossens2003priority}]
    \label{thm:baruah-goossens-funk}
    Let $\pi$ denote a uniform multiprocessor platform as in Lemma
    \ref{lm-uniform-workload}. Let $i$ denote a set of jobs that is
    feasible in $\pi$. Let $\pi'$ an identical multiprocessor platform
    comprised of $m$ unit-speed processors. If Equation
    (\ref{eq:uniform-basic}) is satisfied, then $I$ will meet all
    deadlines when scheduled using G-EDF on~$\pi'$.
  \end{theorem}
}

In practice, according to GFB, a set of  periodic or sporadic tasks
is schedulable by G-EDF if:
\begin{equation}
  \label{eq:uniform-umax}
  U \leq m - (m - 1) U_{\mathsf{max}}
\end{equation}
where  $U_{\mathsf{max}} = \max_i \{ U_i \}$.

The maximum utilisation we can achieve depends on the maximum
utilisation of all tasks in the system: the largest is
$U_{\mathsf{max}}$, the lower is the total achievable
utilisation. This test is only sufficient: if Equation
(\ref{eq:uniform-umax}) is not verified, the task set can still be
schedulable by G-EDF.

The authors of \cite{goossens2003priority} also proposed to give
higher priority to largest utilisation tasks. In this paper we will
not consider these further enhancements.

The BCL test was developed for sporadic tasks, and here we adapt the
notations to server context. We focus on the schedulability of a
target server $S_k$; particularly, we choose one arbitrary server job
of $S_k$.  Execution of the target server job may be interfered by
jobs from other servers with earlier absolute deadlines. The
\emph{interference} over the target server job by an interfering
server $S_i$, within a time interval, is the cumulative length of
sub-intervals such that the target server is in \activeContending\ but
cannot execute, while $S_i$ is running.

%
A \emph{problem window} is the time interval that starts with the
target server job's arrival time and ends with the target server job's
deadline.  As a result, the interference from an interfering server
$S_i$ is upper bounded by its \emph{workload}, which is the cumulative
length of execution that $S_i$ conducts within the problem window.
Let us denote the worst-case workload of a server $S_i$ in the problem
window as $\hat{W}_{i,k}$.

\reportOnly{
  The worst-case workload of a server $S_i$ in the problem window is
  realised in the following scenario (which is further depicted in Figure \ref{fig:worst-ci}) :
  \begin{itemize}
  \item Absolute deadline of the last server job of $S_i$
    in the problem window is coincident with $b$.
  \item All server jobs of $S_i$ are released as soon as possible with
    the minimum activation interval $P_i$.
  \item The \emph{carry-in} server job (i.e., the first job in the
    problem window) executes as late as possible and finishes exactly at
    its absolute deadline.
  \end{itemize}
  By considering this worst-case scenario, an upper bound for the
  workload of $S_i$ within the problem window is:
  \begin{equation}
    \label{eq:upper-bound}
    \hat{W}_{i,k} = \left\lfloor \frac{P_k}{P_i} \right\rfloor\cdot Q_i + \min\{Q_i, P_k\mbox{ mod } P_i\}.
  \end{equation}
}

The formulation of the workload used in this paper is the same as the
one proposed in \cite{bertogna2005improved}.  In order not to
compromise the schedulability when reclaiming CPU time
\reportOnly{(see Section
  \ref{sec:proof-correctness})}\paperOnly{(see~\cite{tech-report-mgrub})},
we need to add one additional term to this upper bound to take into
account the interference caused by the reclaimed bandwidth by servers
that may be activated aperiodically.

Thus, in this paper the workload upper bound is defined as follows:
\begin{equation}
  \label{eq:upper-bound-recl}
  \hat{W}_{i,k} = \left\lfloor \frac{P_k}{P_i} \right\rfloor Q_i + \min\{Q_i, \Delta\} + \max\{\Delta-Q_i, 0\} U_i
\end{equation}
where $\Delta=(P_k\mbox{ mod } P_i)$. 

On the other hand, when $S_i$ and $S_k$ execute in parallel on
different processors at the same time, $S_i$ does not impose
interference on $S_k$. Thus, in case $S_k$ is schedulable, the
interference upon the target job by $S_i$ cannot exceed $(P_k-Q_k)$.

In the end, according to the formulation of BCL used in this paper a
task set is schedulable if one of the following two conditions is true:
\begin{equation}
  \begin{aligned}
    a) \quad \sum_{i\neq k} \min\{\hat{W}_{i,k}, P_k-Q_k\} < m(P_k-Q_k) &\\
    b) \quad \sum_{i\neq k} \min\{\hat{W}_{i,k}, P_k-Q_k\} = m(P_k-Q_k) & \\
              \wedge \neg\exists h\neq k,  \hat{W}_{h,k} \leq P_k-Q_k &
  \end{aligned}
  \label{eq:bcl-admission}
\end{equation}

\reportOnly{
  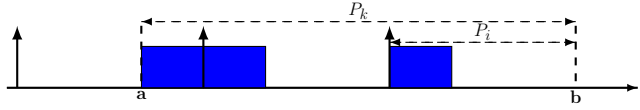
\begin{figure}
    \centering  
    \begin{tikzpicture}[scale=0.275,transform shape]
      \tikzstyle{every node}=[font=\Huge]
      \draw [-latex](-0.5,0) coordinate(dd)-- (0,0) coordinate (O1) -- (30,0)coordinate(ff) node[above]{};
      \draw [thick] (dd) -- (6,0) coordinate (l1) -- (6,0) coordinate (l2) -- (9,0) coordinate (l3) -- (18,0) coordinate (l4) -- (ff);
      
      
      
      \begin{scope}[shift={(l1)}]
        \node[ above right,right,draw, minimum width=3cm,minimum height=2.cm,fill=blue,yshift=1cm](n3a) {};
      \end{scope}
      \begin{scope}[shift={(l3)}]
        \node[ yshift=1cm,right,draw, minimum width=3cm,minimum height=2.cm,fill=blue](n3a) {};
      \end{scope}
      \begin{scope}[shift={(l4)}]
        \node[ yshift=1cm,right,draw, minimum width=3cm,minimum height=2.cm,fill=blue](n3a) {};
      \end{scope}
      \draw[-latex,thick] (O1) -- ++(0,1*3)node[right]{};
      \draw[-latex,thick] (9,0) -- (9,1*3)node[right]{};
      \draw[-latex,thick] (18,0) -- (18,1*3)node[right]{};
      
      \draw[-latex,dashed] (18,2.2)--(27,2.2)node[above]{};
      \draw[-latex,dashed] (27,2.2)--(18,2.2)node[above]{};
      \draw[dashed, thick] (22.5,2.2) node[above]{$P_i$ };
      
      \draw[-latex,dashed] (6,3.2)--(27,3.2)node[above]{};
      \draw[-latex,dashed] (27,3.2)--(6,3.2)node[above]{};
      \draw[dashed, thick] (16.5,3.2) node[above]{$P_k$ };
      

      
      \foreach \xx in{6}{
        \draw[dashed, thick] (\xx,3.) -- (\xx,-0.0) node[below]{$\mathbf{a}$ };
      }
      \foreach \xx in{27}{
        \draw[dashed, thick] (\xx,3.) -- (\xx,-0.0) node[below]{$\mathbf{b}$ };
      }
      
      {
        
        
        
      }
    \end{tikzpicture}
    \caption{Maximum workload by $S_i$ within the problem window}
    \label{fig:worst-ci}
  \end{figure}
}

Between the two tests presented so far (GFB and BCL) no one dominates
the other: there are task sets that are schedulable by GFB but not by
BCL, and vice versa. In general terms, BCL is more useful when a task
has a large utilisation, whereas GFB is more useful for a task set
with many small tasks.


\section{Reclaiming rules}
\label{sec:reclaiming-rule}

In this section we propose two new reclaiming rules for
\mbox{G-EDF}. The first one, that we call \emph{parallel reclaiming}
equally divides the reclaimed bandwidth among all executing
servers. The second one, that we call \emph{sequential reclaiming}
assigns the bandwidth reclaimed from one server to one specific
processor.

Note that, when bandwidth reclaiming is allowed, served jobs within a
server may run for more than the server's budget, as the bandwidth
from other servers may be exploited.
\paperOnly{Due to space constraints, the proofs of correctness are not
  reported here. They can be found in \cite{tech-report-mgrub}.}

\subsection{Parallel reclaiming}
\label{sec:parallel-recl}

In parallel reclaiming, we define one global variable
$U_{\mathsf{inact}}$, initialized to 0, that contains the total
amount of bandwidth in the system that can be reclaimed.  The rules
corresponding to transitions $1$ and $5$ \reportOnly{in Figure
  \ref{fig:dia-mgrub}} are modified as follows.

\begin{description}
\item[5.] A server remains in \activeNonContending\ only if $q_i <
  (d_i - t) U_i$. When $q_i \geq (d_i - t)U_i$ the server moves to
  \inactive. Correspondingly, variable $U_{\mathsf{inact}}$ is
  incremented by $U_i$.
\item[1.] When a job of a task arrives, if the corresponding server is
  \inactive, it goes to \activeContending\ and its budget and deadline
  are modified as in the original rule. Correspondingly,
  $U_{\mathsf{inact}}$ is decremented by $U_i$.
\end{description}

While a server $S_i$ executes on processor $p$, its budget is updated
as follows:

\begin{equation}\label{eq:mgrub-rule-parallel}
  \d{q_i} = -\max\left\{ U_i, 1 - \frac{U_{\mathsf{inact}}}{m} \right\} \d{t}.
\end{equation}

This rule is only valid for the GFB test. That is, if a set of servers
are schedulable by GFB without bandwidth reclaiming, it is still
schedulable when parallel reclaiming is allowed.

\paragraph{Initialization of $U_{inact}$}
While it is safe to initialise $U_{\mathsf{inact}}$ to be 0, we would
like to take advantage of the initial free bandwidth in the system.
Therefore, we initialise $U_{\mathsf{inact}}$ to the maximum initial
value that can be reclaimed without jeopardizing the existing servers.
From Equation (\ref{eq:uniform-umax}), we have:
\begin{equation*}
  U + U_{\mathsf{inact}}  \leq m - (m-1)U_{\mathsf{max}}.
\end{equation*}
That is,
\begin{equation}
  \label{eq:parallel-init}
  U_{\mathsf{inact}}  = m - (m-1)U_{\mathsf{max}} - U.
\end{equation}
This is equivalent to having one or more servers, whose cumulative
bandwidth is $U_{\mathsf{inact}}$ that are always inactive.

\subsection{Sequential reclaiming}
\label{sec:sequential-recl}

In sequential reclaiming, we define an array of variables
$U_{\mathsf{inact}}[]$, one per each processor. The variable
corresponding to processor $p$ is denoted by $U_{\mathsf{inact}}[p]$.
More specifically, $U_{\mathsf{inact}}[p]$ is the reclaimable
bandwidth from inactive servers that complete their executions in
processor $p$, and $U_{\mathsf{inact}}[p]$ can only be used by a
server running on $p$.  For any $p$, $U_{\mathsf{inact}}[p]$ can be
safely initialised to be 0. Then, we modify the rules corresponding to
transitions $1$ and $5$ \reportOnly{in Figure \ref{fig:dia-mgrub}} as
follows.

\begin{description}
\item[5.] A server remains in \activeNonContending\ only if $q_i <
  (d_i - t) U_i$. When $q_i \geq (d_i - t)U_i$ the server moves to
  \inactive. Correspondingly, one of the variables
  $U_{\mathsf{inact}}[p]$ is incremented by $U_i$. The server
  remembers the processor where its utilisation has been stored, so
  that it can recuperate it later on.

\item[1.] When a job of a task arrives, if the corresponding server is
  \inactive, it goes in \activeContending\ and its
  budget and deadline are modified as in the original
  rule. Correspondingly, $U_{\mathsf{inact}}[p]$ (where $p$ is the
  processor where the utilisation was stored before) is decremented by~$U_i$.
\end{description}

While a server $S_i$ executes, its budget is updated
as follows:

\begin{equation}\label{eq:mgrub-rule-sequential}
   \d{q_i} = -\max \left\{U_i, 1 - U_{inact}[p] \right\} \d{t}.
\end{equation}

Notice that, for the moment, we do not explore more sophisticated
methods for updating $U_{\mathsf{inact}}[p]$ when a server becomes
inactive. In fact, there are several possible choices: for example, we
could use a Best-Fit algorithm to concentrate all reclaiming in the
smallest number of processors, or Worst-Fit to distribute as much as
it is possible the reclaimed bandwidth across all processors. In the
current implementation, for simplicity we chose to update the variable
$U_{\mathsf{inact}}[p]$ corresponding to the processor where the task
has been suspended.

This rule works for both the GFB test of Equation
(\ref{eq:uniform-umax}) and and for the modified BCL test of Equation
(\ref{eq:bcl-admission}).

\paragraph{Initialization of $U_{inact}[]$}
Similarly to the parallel reclaiming case, we would like to initialise
each $U_{inact}[p]$ to be an as large as possible value so to reclaim
the unused bandwidth in the system.  Let us denote this value as
$U_x$.

In case the GFB test is used, the maximum free bandwidth is computed
as in Equation \eqref{eq:parallel-init}. Then, to keep the set of
servers still schedulable w.r.t. GFB, we can initialise each variable
to a value $U_x' = \frac{m - (m-1)U_{max} - U}{m}$.

When it comes to the BCL test, we can think of adding $m$ servers to
the system, each one with infinitesimal period and bandwidth equal to
$U_{x}''$.  To allow each server to use free bandwidth as much as
possible while still guaranteeing the schedulability, the following
condition should hold.

\begin{align*}
  \forall k, & \sum_{i \neq k} min(\hat{W}_{i,k}, P_k - Q_k) + m U_x'' P_k
  < m (P_k - Q_k)  \\ 
  U_x'' & < \min_k \left\{ \frac{P_k - Q_k}{P_k} - \frac{\sum_{i \neq k} min(\hat{W}_{i,k}, P_k - Q_k)}{mP_k} \right\}
\end{align*}

Finally we take the maximum between $U_x'$ and $U_x''$, since only one
of the two test needs to be verified.
\begin{equation}
  \label{eq:extra-sequential}
  U_{inact}[p] = \max \left\{U_x', U_x'' \right\}.
\end{equation}


\reportOnly{\section{Proof of correctness}
\label{sec:proof-correctness}

In this section, we prove that the parallel (sequential) reclaiming
algorithm has the temporal isolation property w.r.t. GFB (BCL),
respectively.

When doing reclaiming, we use a different decomposition in server
jobs. Following the technique used in \cite{lipari2000resource} and in
\cite{pellizzoni2008m}, we suppose that every inactive server
continues to generate small server jobs. In the following, we call
these jobs \emph{micro-jobs}: their execution time is used by the
other active servers for free without consuming their respective
budgets.

It is important to stress that, in the description of the reclaiming
rules, these micro-jobs are not necessary, nor they are needed in the
implementation. They are useful only for the proof of correctness.

\subsection{Proof of parallel reclaiming}
\label{sec:proof-para-recl}

Before proceeding to enunciate the temporal isolation theorem, we will
give some preliminary definition and lemmas. 

We start by doing the equivalence between a server $S_i$ with
bandwidth $U_i$ and a dedicated uniform processor of speed $U_i$. 

\begin{definition}
  The virtual time $V_i(t)$ of a server $S_i$ with bandwidth $U_i$ is
  the time at which a dedicated uniform processor of speed $U_i$ has
  executed the same amount of work that the server has executed by
  time $t$ in the shared processor.
\end{definition}

There are several ways to represent the virtual time because when the
server is idle there are several time instant where the uniform
processor has completed its work as well. We chose the following definition:
\begin{itemize}
\item The virtual time is always equal to $t$ while the server is
  inactive, $V(t) = t$.
\item It progresses at a speed $1/U_i$ while the server executes
  $dV_i(t) = dt / U_i$.
\item It does not change while the server is active but not executing.
\end{itemize}
This definition of virtual time can also be formalised as follows:
\begin{equation}
  \label{eq:vtime}
  V_i(t) = \left \{ \begin{aligned}  
      d_i - q_i/U_i   &   \text{if} \;\;S_i \;\;\text{is \activeContending\ or \activeNonContending} \\
      t              &   \text{if} \;\;S_i \;\;\text{is \inactive} 
      \end{aligned} \right .
\end{equation}

In Figure \ref{fig:vtime-evolution} we show the evolution of the
virtual time in the case of one server with parameters $Q_i=2, P_i=4$,
serving a periodic task $\tau_i$ with 4 jobs that sometimes execute
for 1 units of computation, and sometimes for 2 units of
computation. Other examples can be found in
\cite{conf/ecrts/LipariB00} and \cite{lipari2000resource}.

Hence, to execute $Q_i$ units of budget, the uniform processor needs
$P_i$ units of time. However, the server is subject to scheduling, so
it may execute or not, whereas the uniform processor is dedicated, so
it is always executing as long as there are unfinished jobs. If the
task executes less than $Q_i$ every period, there is some slack time
that we may recuperate. This corresponds to the coloured areas in
Figure \ref{fig:vtime-evolution} where the virtual time progresses at
the same speed as the time.

\begin{figure}
  \centering
  \includegraphics[width=10cm]{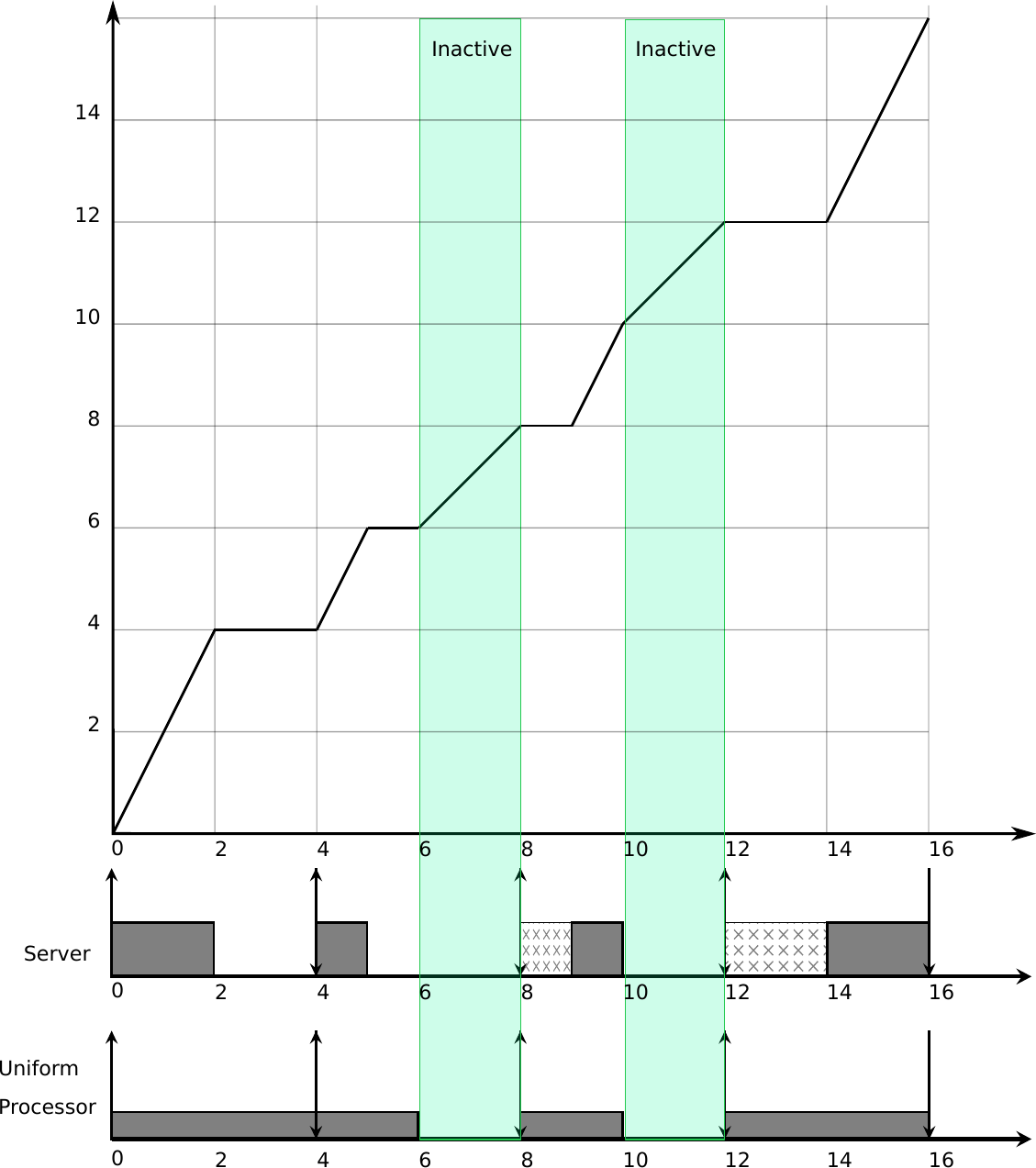}
  \caption{Virtual-time evolution compared to cumulative execution
    time for an example task scheduled by a server with $Q_i = 2$ and
    $P_i=4$.}
  \label{fig:vtime-evolution}
\end{figure}

The equivalence between server and uniform processor can be expressed
in terms of workload.
\begin{lemma}
  \label{lm:vtime-workload}
  Let $I_i$ be the set of jobs generated by server $S_i$; let $\pi'$
  be the shared identical multiprocessor platform, and let $\pi_i$ the
  uniform processor of speed $U_i$. The following relation holds:
  \begin{equation}\label{eq:lemma-workload}
    \forall t, W(EDF, \pi', I_i, t) = W(A, \pi_i, I_i, V_i(t))
  \end{equation}
\end{lemma}
\begin{proof}
  It follows directly from the definition of virtual time and the
  description of the server states in Section \ref{sec:cbs-grub}.
\end{proof}
\begin{lemma}
  When a server becomes \inactive, all jobs have completed both in
  $\pi'$ and in $\pi$.
\end{lemma}
\begin{proof}
  It follows directly from Lemma \ref{lm:vtime-workload}.
\end{proof}

Therefore, $\pi_i$ remains idle until the time when a new job of task
$\tau_i$ is activated. Let $t_1$ denote an instant in which the server
(and the corresponding uniform processor) becomes \inactive, and let
$t_2$ the subsequent instant when $S_i$ becomes active again. Since
the server is idle, it does not consume any bandwidth so we can
reclaim it and donate to another server. However, at time $t_1$ we do
not know yet when in the future the next job will be activated, so we
have to be careful in donating such extra bandwidth.

We divide the interval $[t_1, t_2]$ in a sequence of sub-intervals of 
infinitesimally small length 
$\delta$: $\{[t_1, t_1+\delta), [t_1+\delta, t_1 + 2\delta),
[t_1+2\delta, t_1 + 3\delta), \ldots, [t_2 - \delta, t_2]\}$. As time passes, 
we donate the infinitesimal $\delta U_i$ to the needing server. In this way, 
we can stop as soon as the new job is activated without consuming its 
reserved budget. To model this donation, we introduce the concept of 
\emph{micro-jobs}.  

Each such sub-interval represents the arrival time and the absolute
deadline of a set of \emph{micro-jobs} produced by the inactive
server. More specifically, for sub-interval $[t, t+\delta)$ we produce
$m$ micro-jobs, each one with arrival time at $t$ and absolute
deadline at $t+\delta$, and execution time equal to $\frac{U_i
  \delta}{m}$ so that their cumulative execution time is $U_i \delta$.

It is easy to see that, if we introduce such micro-jobs on the
corresponding dedicated uniform processor, they will all complete
before their deadline. Also, the shape of the virtual time will not
change. To better understand how it works, consider again the example
in Figure \ref{fig:vtime-evolution}: in the coloured intervals the
server is inactive, and the dedicated processor is idle. In these
intervals, we introduce a sequence of micro-jobs, whose cumulative
utilisation is equal to $U_i$. Clearly, since the uniform processor is
idle, it could run all these micro-jobs, and each one would complete
before the deadline. Therefore, the workload executed by $\pi_i$ by
time $t$, including the normal server jobs and all the micro-jobs is
always exactly equal to $U_i t$. In other words, the micro-jobs are
used to \emph{fill the idle times} that are produced on the
\emph{uniform processor} by the fact that a job terminates earlier
than expected. We can repeat this reasoning for every \inactive\
server (uniform processor, respectively). We have at each instant $m$
micro-jobs for each inactive server, and the cumulative execution time
of all the micro-jobs produced by all servers is $\delta
U_{\mathsf{inact}}$.

Under G-EDF, these micro-jobs have maximum priority, since their
relative deadlines are of infinitesimal length. We have at most $m$
micro-jobs active at every instant, so they execute on every processor
for a length of $\delta U_{\mathsf{inact}}$, then they suspend
themselves until the next activation at time $t+\delta$. Thus, in an
interval of length $\delta$, the $m$ earliest deadline ``normal
servers'' will execute consuming their budget for only $(1 -
U_{\mathsf{inact}}/m)\delta$, whereas they can execute without
consuming their budget (which is accounted for by the micro-jobs) for
the remaining $\delta U_{\mathsf{inact}}/m$.

With these observation in mind, we can now proceed to prove the main
theorem.

\begin{theorem}\label{thm:parall-recl-proof}
  If Equation (\ref{eq:uniform-umax}) holds, then all jobs produced by
  all servers under parallel reclaiming will meet their scheduling
  deadlines.
\end{theorem}

\begin{proof}  
  Equation (\ref{eq:uniform-umax}) is still respected after adding the
  micro-jobs (the utilisation of the servers has not changed),
  therefore Lemma \ref{lm-uniform-workload} and Theorem
  \ref{thm:baruah-goossens-funk} are still valid and the new
  collection of jobs, including the micro-jobs, is schedulable by EDF
  on the identical multiprocessor platform.
  
  Observe that we can take any small $\delta > 0$; hence the $m$
  micro-jobs that are active at time $t$ have the highest priority in
  the identical multiprocessor platform: they execute immediately
  after their activation for an interval of time equal to $\delta
  U_{\mathsf{inact}} / m$. Therefore, the executing servers at
  time~$t$ can execute for $\delta U_i / m$ without consuming their
  budget, and for $\delta (1 - U_j / m)$ consuming their budget. In
  any case, a server must always consume all its budget within a
  period~$P_i$, so the budget decreasing rate of a server cannot be
  less than $U_i$.  Therefore, the budget of the servers can be
  updated in Equation (\ref{eq:mgrub-rule-parallel}), without
  compromising the schedulability of the servers.
\end{proof}

Unfortunately, parallel reclaiming does not work for the BCL test. The
reason is that, by dividing all spare bandwidth in parallel pieces,
the interference caused to a server may increase too much. 

%
%

\subsection{Proof for sequential reclaiming}
\label{sec:proof-sequ-recl}

The proof of correctness of the sequential reclaiming rule for the GFB
test is similar to the proof presented in the previous section: it
suffices to create one single micro-job for each reclaiming server. In
particular, every \inactive\ server in interval $[t, t+\delta]$
generates one single micro-job with budget $\delta U_i$. Under G-EDF
this micro-job will execute with the highest priority on any
processor, so we fix it on one specific processor. The proof follows
in the same way as the proof of Theorem \ref{thm:parall-recl-proof}.

Before proceeding with the proof for BCL, we need to discuss the
maximum workload $\hat{W}_{i,k}$ for a server $S_i$ in the problem
window. Moreover, we would like to differentiate two cases.
\begin{enumerate}
\item Periodic servers that are always periodically activated by its
  served tasks. 
\item Servers that may be activated aperiodically.
\end{enumerate}

For case (1), it is easy to see that, even with bandwidth reclaiming,
the worst-case workload of a server $S_i$ in the problem window is
realised as in the scenario depicted in Figure \ref{fig:worst-ci}.
This can be checked such that by violating any condition as in Figure
\ref{fig:worst-ci}, the workload will not increase. The scenario is the following:
\begin{itemize}
\item Absolute deadline of the last job of $S_i$ in the problem window
  is coincident with $b$.
  \item All server jobs of $S_i$ are periodically released according to the server period $P_i$.
  \item The \emph{carry-in} server job (i.e., the first job in the problem window) executes
    as late as possible and finishes exactly at its absolute deadline.
\end{itemize}

However, when it comes to case (2) such that servers in the system
admit aperiodic activation, the scenario in Figure \ref{fig:worst-ci}
may not correspond to the maximum workload.  We provide here a sketch
of the reasoning.

Let define $\Delta=(P_k\mbox{ mod } P_i)$ and suppose $\Delta > Q_i$.
This means that we can anticipate (move to the left) the chain of
arrivals in Figure \ref{fig:worst-ci} as much as $(\Delta-Q_i)$ time
units without reducing the workload. Moreover, if a server allows
aperiodic activation, then from time $b-(\Delta-Q_i)$ to time $b$,
server $S_i$ will generate the same micro-jobs as in the proof of
Theorem \ref{thm:parall-recl-proof}, which can at most contribute
$(\Delta-Q_i) U_i$ to the workload . In the end, for aperiodic servers
with bandwidth reclaiming we formulate $S_i$'s workload by Equation
\eqref{eq:upper-bound-recl}.  More specifically, comparing Equation
\eqref{eq:upper-bound} and \eqref{eq:upper-bound-recl}, the additional
term $(\Delta-Q_i)$ can be regarded as the maximum delay a server job
of $S_i$ within the problem window can experience such that $S_i$'s
workload keeps monotonically increasing.

We now prove that, by admitting a server using the test of Equation
\eqref{eq:bcl-admission} and then performing sequential reclaiming,
no server job misses its deadline.

\begin{theorem}\label{thm:seq-recl-proof}
  Equation (\ref{eq:bcl-admission}) still holds for a set of servers
  with sequential reclaiming.
\end{theorem}

\begin{proof}
  This proof will deal with the general case and servers may admit
  aperiodic activations.
  
  We first prove that the formulation in Equation
  \eqref{eq:upper-bound-recl} for $\hat{W}_{i,k}$ is indeed the
  maximum workload. And for simplicity, let us assume that
  $\Delta>Q_i$ (if $\Delta \leq Q_i$ the reasoning is the same).
  
  We define $\phi$ to represent the amount of carry-in workload from
  $S_i$, that is, $\phi=\min\{Q_i,P_k\mbox{ mod } P_i\}$.  And $W'$
  denotes the amount of workload from execution of non-carry-in jobs
  of $S_i$ within the problem window such that these jobs have
  absolute deadlines no later than time point $b$; let say there are
  $n'$ such jobs, that is, $n'=\lfloor\frac{P_k}{P_i}\rfloor$.  Then,
  the maximum workload by micro-jobs from $S_i$ (that is, workload
  caused by donating bandwidth to other servers' execution) would be
  $(n' P_i-\frac{W'}{U_i}) U_i + (\Delta-Q_i) U_i = n' P_i U_i-W' +
  (\Delta-Q_i) U_i$.  Remind that $(\Delta-Q_i) U_i$ is the part of
  workload resulted from $S_i$'s aperiodic behaviour.
    
  The total workload due to $S_i$ within the problem window is
  $\phi+W' + (n' P_i U_i - W')+ (\Delta-Q_i) U_i= \phi + n' Q_i +
  (\Delta-Q_i) U_i = \hat{W}_{i,k}$.  As we can see, the interference
  upon the target job due to workload from~$S_i$ can be upper bounded
  by $\hat{W}_{i,k}$.
    
  On the other side, a server can donate its budget only when it is
  inactive and at any time its budget can be used by at most one
  executing server (according to the sequential reclaiming rule).
  That is, the execution of $S_i$ and its budget donation must be in a
  sequential way. Thus, if $S_k$ is schedulable, the interference due
  to $S_i$ within the problem window can still be upper bounded by
  $(P_k-Q_k)$.
    
  In the end, Equation (\ref{eq:bcl-admission}) still holds if
  sequential reclaiming is applied.
\end{proof}

}

\paperOnly{\section{Implementation}
\label{sec:implementation}

The parallel and sequential reclaiming techniques described in the
previous sections have been implemented in the Linux kernel, extending
the {\tt SCHED\_DEADLINE} scheduling policy~\cite{SPE:SPE2335}.  The
modified Linux kernel has been publicly released at
\url{https://github.com/lucabe72/linux-reclaiming}.  These kernel
modifications are based on a previous implementation of the GRUB
algorithm~\cite{Abe14RTLWS}, which however did not guarantee
schedulability of server jobs.

\subsection{Parallel reclaiming implementation}
Parallel reclaiming requires to keep track of the total inactive
bandwidth in a global (per-root domain) variable $U_{\mathsf{inact}}$
which is updated when tasks move from an {\sf Active} states to the
\inactive\ state or vice-versa.

Then, the budget decrease rate of every executing server depends on
this global variable (see Equation
(\ref{eq:mgrub-rule-parallel})). For each executing {\tt
  SCHED\_DEADLINE} task, the scheduler periodically accounts the
executed time, decreasing the current budget (called ``runtime'' in
the Linux kernel) of the task at each tick (or when a context switch
happens).  When a reclaiming strategy is used, the amount of time
decreased from the budget depends on the value of the global variable
$U_{\mathsf{inact}}$. This means that, when a server changes its state
to \inactive\ (or \activeContending) and the value of
$U_{\mathsf{inact}}$ changes all the CPUs should be signalled to
update the budgets of the executing task before $U_{\mathsf{inact}}$
is changed. Such an inter-processor signalling can be implemented
using Inter-Processor Interrupts (IPIs). However, this may
substantially increase the overhead of the scheduler and increase its
complexity; furthermore the combination of global variables and
inter-processor interrupts may lead to race conditions very difficult
to identify.

Therefore, parallel reclaiming has been implemented by introducing
a small approximation: we avoid IPIs and the value of
$U_{\mathsf{inact}}$ is sampled only at each scheduling tick.
In this regard it is worth noting that every real scheduler
implements an approximation of the theoretical scheduling algorithm:
for example, {\tt SCHED\_DEADLINE} accounts the execution time at every
tick (hence, a task can consume up to 1 tick more than the reserved
runtime / budget).

Despite of the approximations introduced when implementing parallel
reclaiming, during our experiments with randomly generated tasks we
never observed any server deadline miss, probably because the GFB
schedulability test is pessimistic and hence a certain amount of slack
is available in the great majority of cases.
It is important to underline, however, that from a purely theoretical
point of view our current implementation of the parallel reclaiming
rule cannot guarantee the respect of every server deadline.

\subsection{Sequential reclaiming implementation}
Implementing sequential reclaiming is easier under certain
assumptions. In particular, we need to make sure that the code
executing on the $p$-th runqueue only accesses variables local to the
same runqueue.

In sequential reclaiming, we need to provide one variable
$U_{\mathsf{inact}}$ for each runqueue. When a server on the $p$-th
runqueue becomes \inactive, we update the corresponding variable; at
this point, only the budget of the task executing on this runqueue
needs to be updated. When a server becomes {\sf Active}, we make sure that
the corresponding handler is executed on the same processor where the
server was previously suspended and became \inactive\ (the task will
be migrated later, if necessary). Therefore, we just need to modify
the local $U_{\mathsf{inact}}$ and update the budget of the task
executing on this CPU.  While this may not be the optimal way to
distribute the spare bandwidth, we do not need any IPI 
to implement the exact reclaiming rule (the only approximations
are the ones introduced by the {\tt SCHED\_DEADLINE} accounting
mechanism).

Notice that for both parallel and sequential reclaiming the transition
between \activeNonContending\ and \inactive\ is handled by setting up
an {\em inactive timer} that fires when such a transition should
happen (for more details, see~\cite{Abe14RTLWS}, where such a time is
named {\em 0-lag time}).

\section{Experimental Evaluation}
\label{sec:experiments}
The effectiveness of the reclaiming algorithms has been evaluated
through some experiments with the patched Linux kernel described in
the previous section.  The kernel version used for all the experiments
is based on version 3.19 (in particular, the {\tt global-reclaiming}
and {\tt refactored-reclaiming} branches of the {\tt linux-reclaiming}
repository have been used).  All the tests were executed on a
$4$-cores Intel Xeon CPU.


\subsection{Randomly generated tasks}
\label{sec:rand-gener-tasks}
The first set of experiments has been performed by executing sets of
randomly generated real-time tasks with the {\tt rt-app}
application\footnote{\url{https://github.com/scheduler-tools/rt-app}}.

A set of $100$ task sets with utilisation $U = 2.5$ has been generated
using the {\tt taskgen}~\cite{Emberson10-WATERS} script, and the task
sets that are schedulable according to the BCL and GFB tests have been
identified. Some first {\tt \mbox{rt-app}} runs confirmed that these
task sets can actually run on the Linux kernel (using {\tt
  SCHED\_DEADLINE}) without missing any deadline.
Then, the reclaiming mechanisms have been tested as follows: for each
schedulable task set $\gamma=\{(C_i, T_i)\}$ generated by {\tt
  taskgen}, a task set $\Gamma=\{\tau_i'\}$ has been generated, where
$\tau_i'$ has period $T_i$ and execution time uniformly distributed
between $\alpha \gamma C_i$ and $\gamma C_i$ (hence, the WCET of task
$\tau_i$ is $\gamma C_i$), and is scheduled by a server with
parameters $(Q_i = C_i, P_i = T_i)$.  Notice that $\gamma$ represents
the ratio between the task's WCET and the maximum budget allocated to
the task; hence increasing $\gamma$ increases the amount of CPU time
that the task needs to reclaim to always complete before its deadline;
on the other hand, $\alpha$ represents the ratio between the BCET and
the WCET of a task (so, $\alpha \leq 1$). Decreasing $\alpha$
increases the amount of CPU time that a task can donate to the other
tasks through the reclaiming mechanism.  When $\gamma \leq 1$, the
WCET of each task is smaller than the maximum budget $Q_i$ used to
schedule the task, so all the tasks' deadlines will be respected. The
experiments confirmed this expectation.  When $\gamma > 1$, instead,
the situation is more interesting because some deadlines can be missed
and enabling the reclaiming mechanism allowed to reduce the amount of
missed deadlines.

\begin{figure}[t]
  \begin{center}
    \includegraphics[width=.9\columnwidth]{Results/Exp1/exp1.pdf}
  \end{center}
  \vspace{-5mm}
  \caption{Percentage of missed deadlines when using {\tt rt-app} with
           different reclaiming strategies. $\alpha$ varies from $0.2$ to
           $0.8$, and $\gamma = 1.1$.}
  \label{fig:exp1}
\end{figure}
Figure~\ref{fig:exp1} reports the percentage of missed deadlines for
$\gamma=1.1$ as a function of $\alpha$ when using no reclaiming,
parallel reclaiming and sequential reclaiming. For parallel and
sequential reclaiming, the results are reported both when initialising
$U_{inact}$ to $0$ and when initialising it according to
Equations~\eqref{eq:parallel-init} and~\eqref{eq:extra-sequential}
(reclaiming the initial spare utilisation).  From the figure, it can
be seen that both reclaiming algorithms allow to reduce the percentage
of missed deadlines; however, parallel reclaiming tends to perform
better than sequential reclaiming. When $\alpha$ increases, the
average utilisation of the tasks increases, and the amount of CPU time
that can be reclaimed decreases; hence, the differences between the
efficiency of the various algorithms become more evident; however,
parallel reclaiming performance do not seem to depend too much on the
value used to initialise $U_{inact}$. This happens because with a
small value of $\gamma=1.1$, the tasks do not need to reclaim much
execution time.

\begin{figure}[t]
  \begin{center}
    \includegraphics[width=.9\columnwidth]{Results/Exp1/exp13.pdf}
  \end{center}
  \vspace{-5mm}
  \caption{Percentage of missed deadlines when using {\tt rt-app} with
           different reclaiming strategies. $\alpha$ varies from $0.2$ to
           $0.8$, and $\gamma = 1.3$.}
  \label{fig:exp13}
\end{figure}
Increasing the value of $\gamma$ to $\gamma = 1.3$
(Figure~\ref{fig:exp13}), the tasks need to reclaim more execution
time and the effect of $U_{inact}$ initialisation becomes more
evident. In particular, with $\alpha=0.8$ the tasks cannot donate
enough execution time, so if $U_{inact}$ is initialised to $0$ (only
the utilisation of the ``existing tasks'' can be reclaimed) the two
reclaiming algorithms (parallel and sequential) do not seem to be very
effective (the percentage of missed deadlines is similar to the ``No
Reclaiming'' case). If, instead, $U_{inact}$ is initialised according
to Equations~\eqref{eq:parallel-init} and~\eqref{eq:extra-sequential}
the reclaiming algorithms are able to reclaim the spare utilisation
and are able to reduce the percentage of missed deadlines.

Some partial conclusions can be drawn from this set of experiments. In
general, the parallel reclaiming strategy performs better than
sequential reclaiming. This is probably due to the fact that parallel
reclaiming tends to fairly distribute the spare bandwidth across all
processors, whereas in the current implementation of the sequential
reclaiming we have no control on which server uses the reclaimed
bandwidth. In fact, with sequential reclaiming, in the worst-case all
reclaiming could go on one single processor and benefit only the tasks
that by chance execute on that processor. 

On the other hand, using sequential reclaiming we can admit a larger
number of tasks sets, because the mechanism is valid both for the GFB
and a modified version of the BCL test. Furthermore, as previously
discussed a precise implementation of the parallel reclaiming is more
costly in terms of overhead and programming effort. For the moment we
conclude that sequential reclaiming seems to be preferable from an
implementation point of view. However we acknowledge that further
investigation is needed to perform a full assessment.

\subsection{Experiments on real applications}

In the next set of experiments, the performance of a real application
(the {\tt mplayer} video
player\footnote{\url{http://www.mplayerhq.hu}}) has been evaluated.
In particular, {\tt mplayer} has been modified to measure the
``Audio/Video delay'', defined as the difference of the presentation
timestamps of two audio and video frames that are reproduced
simultaneously.  A negative value of the Audio/Video delay means that
video frames are played in advance with respect to the corresponding
audio frames, while a positive value indicates that the video is late
with respect to the audio (probably because {\tt mplayer} has not been
able to decode the video frames in time).  When this value becomes too
large, audio and video are perceived out of synch, and the quality
perceived by the user is badly affected.

When {\tt mplayer} is executed as a {\tt SCHED\_DEADLINE} task, it is
pretty easy to set the reservation period $P=1/fps$, where $fps$ is
the frame rate (in frames per second) of the video; however, correctly
dimensioning the maximum budget/runtime $Q$ is much more difficult. If
$Q$ is slightly under-dimensioned (larger than the average time needed
to decode a frame, but smaller than the maximum time), the Audio/Video
delay can become too large affecting the quality, and a reclaiming
mechanism can help in improving the perceived quality.

\begin{figure}[t]
  \begin{center}
    \includegraphics[width=.9\columnwidth]{Results/MPlayer/av.pdf}
  \end{center}
  \vspace{-5mm}
  \caption{Audio/Video delay experienced by {\tt mplayer} reproducing a video
           when scheduled with {\tt SCHED\_DEADLINE}. The three plots indicate
           {\tt mplayer} executing alone or together with other real-time
           tasks without reclaiming, or with the M-GRUB reclaiming
           mechanism.}
  \label{fig:mplayer}
\end{figure}
For example, Figure~\ref{fig:mplayer} shows the evolution of the
Audio/Video delay experienced by {\tt mplayer} when reproducing a
full-D1 mpeg4 video (with vorbis audio) with $Q=4.5ms$ and $P=40ms$
(the video is $25$ frames per second, so $P=1s / 25 = 40ms$). The
experiment has been repeated executing {\tt mplayer} alone on an idle
4-cores system (``No Reclaiming, idle system'' line) or together with
other real-time tasks (implemented by {\tt rt-app}, in the ``No
Reclaiming, loaded system'' line). In the ``loaded system'' case, the
total utilisation was about $2.2$ and the task set resulted to be
schedulable according to both BCL and GFB. As it can be noticed, in
both cases the Audio/Video delay continues to increase, and becomes
noticeable for the user. When the M-GRUB reclaiming mechanism is
activated, {\tt mplayer} can use some spare time left unused from the
other tasks, and the Audio/Video delay is unnoticeable (see
``Reclaiming, loaded system'').

The experiment has been repeated with parallel and sequential
reclaiming, obtaining identical results. Hence, in this specific case
using one policy instead of the other does not bring any particular
advantage.

Notice that the scheduling parameters (reservation period and maximum
budged) of the {\tt rt-app} real-time tasks have been dimensioned so
that no deadline is missed. During the experiments, it has been
verified that the number of missed deadlines for such tasks is
actually $0$, even when the reclaiming mechanism is enabled. 


}

\section{Conclusions}
\label{sec:conclusions}

In this \paperOnly{paper}\reportOnly{report}, we proposed two
different reclaiming mechanisms for real-time tasks scheduled by G-EDF
on multiprocessor platforms, named \emph{parallel} and
\emph{sequential reclaiming}. \paperOnly{After proving their
  correctness, we described their implementation on the Linux OS, and
  compared their performance on synthetic experiments.  Parallel
  reclaiming requires more approximations in its implementation,
  however, in average it performs better than the sequential
  reclaiming. On the other hand, sequential reclaiming can guarantee
  the real-time schedulability of a large number of tasks, as it
  allows to use a different admission test, and is characterised by a
  simpler implementation.  However, it performs slightly worse in
  average.} In the future we plan to conduct further investigations
comparing the two strategies, and to use more advanced admission
tests.

\bibliographystyle{abbrv}
\bibliography{retis}

\end{document}